\documentclass[11pt,reqno,a4paper,oneside]{amsart}	

\usepackage{amsmath, amssymb, amsthm, amstext}		
\usepackage{txfonts}								
\usepackage[mathletters]{ucs}
\usepackage[utf8x]{inputenc}
\usepackage[english]{babel}
\usepackage{harvard}								
\usepackage[T1]{fontenc}
\usepackage{ae,aecompl}
\usepackage{enumerate}

\usepackage{appendix}
\usepackage{lmodern}
\usepackage{epsfig,color}
\usepackage[arrow, matrix, curve]{xy} 				
\usepackage{tensor} 								
\usepackage{cancel}								

\usepackage{color}								
\usepackage{mathrsfs}
\usepackage{bbm}
\usepackage{comsym}
\usepackage{graphicx}

\usepackage{babelbib}							

\numberwithin{equation}{section}

\newtheorem{lem}{Lemma}[section]
\newtheorem{prop}[lem]{Proposition} 
\newtheorem{thm}[lem]{Theorem}
\newtheorem{cor}[lem]{Corollary}

\theoremstyle{definition}
\newtheorem{Def}[lem]{Definition} 
\newtheorem{rem}[lem]{Remark}
\newtheorem{ex}[lem]{Example}

\newtheorem*{thm*}{Theorem}

\usepackage[textwidth=5in,textheight=8.25in,paperheight=11in]{geometry}
\usepackage{setspace}
\allowdisplaybreaks[1]  
\usepackage{hyperref}
\hypersetup{%
  pdftitle = {Twistors and GR},
  pdfauthor = {Norman Metzner},
  bookmarksnumbered=true,
  bookmarksdepth =1,
  pdfdisplaydoctitle=true,
  linkbordercolor={0 0 1}
}

\input{headings}

\title{Twistor Theory of Higher-Dimensional Black Holes --- Part~I: Theory}
\author[N Metzner]{Norman Metzner$^{\dagger}$}\thanks{$^{\dagger}$Mathematical Institute, University of Oxford, 24-29 St\,Giles', OX1 3LB Oxford, UK, and
St\,John's College, OX1 3JP, Oxford, UK}

\begin{document}


\addtolength{\jot}{0.1cm}			

\begin{abstract}
The correspondence between stationary, axisymmetric, asymptotically flat space-times and bundles over a reduced twistor space has been established in four dimensions. The main impediment for an application of this correspondence to examples in higher dimensions has been the lack of a higher-dimensional equivalent of the Ernst potential. This article will propose such a generalized Ernst potential, point out where the rod structure of the space-time can be found in the twistor picture and thereby provide a procedure for generating solutions to the Einstein equations in higher dimensions from the rod structure and other asymptotic data. An important result for the study of five-dimensional examples will be the theorem which relates the patching matrices on the outer semi-infinite rods.
\end{abstract}

\maketitle   
\tableofcontents

\section{Introduction} \label{sec:intro}
Based on the initial work by \cite{Ward:1983yg}, a correspondence between stationary, axisymmetric, asymptotically flat solutions of the Einstein equations in four 
dimension and bundles over twistor space was established; see for example \cite{Woodhouse:1988ek}, \cite{Fletcher:1990db}, \cite{Fletcher:1990aa}, \cite{Mason:1996hl}. The correspondence 
is based on a symmetry-reduced version of the Penrose-Ward transform (PW in the diagram below) and the Ernst equation, which provides a link between general relativity and anti-self-dual Yang-Mills (ASDYM) theory. The following diagram depicts it in a nutshell
\begin{equation} \label{diagr:maincorr} 
\hspace{-0.15cm}
\begin{xy}
  \xymatrix@!R@C=1.7cm{
      & \parbox{3.2cm}{{ASDYM connection for rank-$(n-2)$ vector bundle $B\to U$}} \ar@{<->}[r]^(0.615){\text{\tiny PW}} \ar[d]_(0.545){\txt{\tiny symmetry \\ \tiny reduction}} &  \parbox[c][0.9cm][c]{1.4cm}{{$B'\to \mathcal{P}$}} \ar[d]_(0.545){\txt{\tiny symmetry \\ \tiny reduction}} \\
      \parbox{2.9cm}{{stationary axi\-sym\-me\-tric space-time of di\-men\-sion $n$}} \ar@{<->}[r]^(0.48){\text{\tiny coincidence}} 	& \makebox[3.4cm][c]{{reduced ASDYM}} \ar@{<->}[r]	& \parbox[c][0.3cm][c]{1.4cm}{{\hspace{0.1cm} $E\to \mathcal{R}$}}     
  }
\end{xy}
\end{equation}

The idea of this article is to use this correspondence as an approach to the problem of black hole classification in five dimensions. Sections~\ref{sec:yang} and \ref{sec:bundles} will 
quickly review some of the intermediate steps in \eqref{diagr:maincorr} and describe the bundle $E\to \mathcal{R}$. We will see that a characterization of this bundle is 
possible by essentially specifying only one transition matrix, the so-called patching matrix. Even though the correspondence works in any number of space-time dimension, 
not all of the practically important structures can be immediately carried over from four to higher dimensions, in particular the Ernst potential. In Section~\ref{sec:fivedim} 
we therefore define a generalized Ernst potential which is obtained from the matrix $J$ (defined just before Equation (\ref{eq:redeinst2})) by a B\"acklund transformation. Furthermore, we point out the relation between the 
twistor bundle and the rod structure. 
The rod structure is a promising candidate for the additional data in black hole classification in higher dimensions. In 
Section~\ref{sec:bh} we review its definition and a visualization, and in Section~\ref{sec:fivedim} we will show where it can be found in the twistor picture. 

By Theorem~\ref{thm:invpmatrix} we will provide a key tool for the reconstruction of the patching matrix, hence the space-time metric, from the parameters mass, angular 
momenta and rod structure. In a follow-up to this article it will be shown how this method can be applied to various examples.

\section{ASDYM meets Einstein} \label{sec:yang}
One of the two steps in the correspondence \eqref{diagr:maincorr} is the coincidence that the ASDYM equations, reduced by symmetries, turn out to 
be equivalent to the Ernst equation for stationary axisymmetric gravitational fields in general relativity. This was originally discovered by \cite{Ward:1983yg,Witten:1979aa}.

From \cite[Sec.~6.6]{Mason:1996hl} and references therein, we learn that in flat space and with respect to the metric 
\begin{equation*}
\drm s^{2} = -\drm t^{2}+\drm z^{2}+\drm r^{2}+r^{2}\drm {\theta}^{2},
\end{equation*}
where the Killing vectors $∂_{t}$ and $∂_{\theta}$ represent a time-translational and rotational symmetry, the ASDYM field equations can be written as
\begin{equation} \label{eq:redyang}
r\partial _{z}(J^{-1}\partial _{z}J)+\partial _{r}(rJ^{-1}\partial _{r}J)=0.
\end{equation}
Here $J=J(r,z)$ is a matrix whose properties depend on the gauge group of the YM theory. Every solution to \eqref{eq:redyang}, which is called the \textit{reduced Yang's equation}, determines a stationary axisymmetric ASDYM field and every stationary axisymmetric ASDYM field can be obtained in that way. The \textit{Yang's matrix} $J$ determines the connection up to $J\mapsto A^{-1}JB$ with constant matrices $A$ and $B$.

\subsection{Reduction of Einstein Equations}
Now the question is how the reduced Yang's equation~\eqref{eq:redyang} arises from a reduction of the Einstein equations. 

Following \cite[Sec.~6.6]{Mason:1996hl}, we let $g_{ab}$ be a metric tensor in $n$ dimensions (real or complex), and $X_{i}^{a}$, $i=0,\ldots ,n-s-1$, be $n-s$ commuting Killing vectors that generate an orthogonally transitive isometry group with non-null $(n-s)$-dimensional orbits. This means the distribution of $s$-plane elements orthogonal to the orbits of $X_{i}$ is integrable, in other words $[U,V]$ is orthogonal to all $X_{i}$ whenever $U$ and $V$ are orthogonal to all $X_{i}$. 

If we define $J=(J_{ij})\coloneqq (g_{ab}^{\vphantom{1}}X_{i}^{a}X_{j}^{b})$, then a longer calculation shows that the Einstein vacuum equations are equivalent to
\begin{equation} \label{eq:redeinst2}
D_{a}(rJ^{-1}D^{a}J)=0,
\end{equation}
where $D_{a}$ is the covariant derivative on the quotient space by the Killing vectors, $\Sigma$, identified with any of the $s$-surfaces orthogonal to the orbits. 
Hence, the indices in \eqref{eq:redeinst2} run over $1,\ldots ,s$ and are lowered and raised with the metric (and its inverse) on $S$. The variable $r$ is defined by $-r^{2}=\det J$ and 
taking the trace of \eqref{eq:redeinst2} leads to $D^{2}r=0$, that is $r$ is harmonic on $\Sigma$. Let us assume from now on that the gradient of $r$ is not null.

Of particular importance to us is the case $s=2$ with Riemannian $\Sigma$, because then we may introduce isothermal coordinates, that is we can write the metric on $\Sigma$ in the form
\begin{equation*}
{\erm}^{2{\nu}}(\drm r^{2}+\drm z^{2})
\end{equation*}
where $z$ is the harmonic conjugate to $r$. As the Killing vectors commute, there exist coordinates $(y_{0},\ldots ,y_{n-3})$, where the $X_{i}$ are the first $n-2$ coordinate vector fields. Taking 
isothermal coordinates for the last two components, the full metric then has the form
\begin{equation} \label{eq:sigmametric}
\drm s^{2} = \sum _{i,j=0}^{n-3}J_{ij}\,\drm y^{i} \drm y^{j} + {\erm}^{2{\nu}} (\drm r^{2}+\drm z^{2}),
\end{equation}
Now \eqref{eq:redeinst2} reduces to \eqref{eq:redyang} and we obtain the following proposition.

\begin{prop}[Proposition 6.6.1 in \cite{Mason:1996hl}.]
Let $g_{ab}$ be a solution to Einstein's vacuum equation in $n$ dimensions. Suppose that it admits $n-2$ independent commuting Killing vectors generating an orthogonally transitive isometry group with non-null orbits, and that the gradient of $r$ is non-null. Then $J(r,z)$ is the Yang's matrix of a stationary axisymmetric solution to the ASDYM equation with gauge group $\mathrm{GL}(n-2,\mathbb{C})$.
\end{prop}

The partial converse of this Proposition yields a technique for solving Einstein's vacuum equations as follows. Any real solution $J(r,z)$ to \eqref{eq:redyang} such that
\begin{enumerate}[(a)]
\item $\det J=-r^{2}$,
\item $J$ is symmetric
\end{enumerate}
determines a solution to the Einstein vacuum equations, because we can reconstruct the metric from $J$ and ${\erm}^{2{\nu}}$ via \eqref{eq:sigmametric}, and then \eqref{eq:redyang} is equivalent to the vanishing of the components of $R_{ab}$ along the Killing vectors. The remaining components of the vacuum equations can be written as
\begin{equation} \label{eq:redeinst3}
2 \irm \partial _{{\xi}}\left(\log\left(r{\erm}^{2ν}\right)\right)=r\mathrm{tr} \left(\partial_{\xi}\left(J^{-1}\right) \partial_{\xi} J\right),
\end{equation}
with ${\xi}=z+\irm r$, together with the complex conjugate equation 
(see \cite[App.~D, Eq.~(D9)]{Harmark:2004rm}). These equations are automatically integrable if \eqref{eq:redyang} is satisfied and under the constraint $\det J=-r^{2}$, and they determine ${\erm}^{{\nu}}$ 
up to a multiplicative constant. The constraint on $\det J$, however, is not significant for the following reason. We know that in polar coordinates $u=c \log r + \log d$ for constants $c$ and $d$ 
is a solution to the (axisymmetric) Laplace equation
\begin{equation} \label{eq:axilap}
\partial_{r}^{\vphantom{2}}(r\partial_{r}^{\vphantom{2}}u)+r\partial _{x}^{2}u=0,
\end{equation}
so suppose $J$ is a solution to \eqref{eq:redyang}, and consider $\erm^{u}J=dr^{c}J$. Substituting this new matrix in \eqref{eq:redyang}, we see that it is again a solution of the 
reduced Yang's equation. The determinant constraint can thus be satisfied by an appropriate choice of the constants, since we have
\begin{equation*}
\det(\erm^{u}J)=\erm^{(n-2)u}\det J=d^{n-2} r^{(n-2)c}\det J.
\end{equation*}
The condition $J=J^{\mathrm{t}}$ is a further $\mathbb{Z}_{2}$ symmetry imposed on the ASD connection. 


Considering the case $n-2=s=2$, we can write the metric in canonical Weyl coordinates
\begin{equation*}
\drm s^{2} = -f(\drm t-{\alpha} \,\drm {\theta})^{2}+f^{-1}r^{2}\,\drm{\theta}^{2}+{\erm}^{2{\nu}}(\drm r^{2}+\drm z^{2}),
\end{equation*}
from which Yang's matrix $J$ can be read off as
\begin{equation} \label{eq:backldec4d}
J=\left(\begin{array}{cc}
-f{\alpha}^2+r^2f^{-1} & f{\alpha} \\
f{\alpha} & -f
\end{array}\right),
\end{equation}
with $f$ and ${\alpha}$ functions of $z$ and $r$. Then the second component of the reduced Yang's equation is an integrability condition for a function ${\psi}$ with 
\begin{equation} \label{eq:intweyl}
r\partial _{z}{\psi}+f^{2}\partial _{r}{\alpha}=0, \quad r \partial _{r}{\psi}-f^{2}\partial _{z}{\alpha}=0. 
\end{equation}
If we now consider the matrix
\begin{equation} \label{eq:ernstpot4d}
J'=\frac{1}{f} \left(\begin{array}{cc}
{\psi}^2+f^2 & {\psi} \\
{\psi} & 1
\end{array}\right),
\end{equation}
we find that $J'$ satisfies \eqref{eq:redyang} if and only if $J$ does. Solutions to Einstein's vacuum equations can 
therefore be obtained by solving \eqref{eq:redyang} for $J'$ subject to the 
conditions $\det J'=1$, $J'=J'^{\mathrm{t}}$. In this context \eqref{eq:redyang} is called the \textit{Ernst equation} and the complex function $\mathcal{E}=f+\irm {\psi}$ is the \textit{Ernst potential} \cite{Ernst:1968aa}, which is often taken as the basic variable in the analysis of stationary axisymmetric fields. We will also refer to $J'$ as the Ernst potential. 
\section{Bundles over Reduced Twistor Space} \label{sec:bundles}
The second link on which the correspondence in Figure~\ref{diagr:maincorr} hinges is the Penrose-Ward transform, by which a 
one-to-one correspondence between solutions of the ASDYM field equations and bundles over twistor space is established, see \cite[Thm.~10.2.1]{Mason:1996hl}. 
This section, which is based on \cite{Fletcher:1990db}, will describe the bundle structure which represents the solution of Einstein's equations. 

First we define the reduced version of the correspondence space as $\mathcal{F}_{\mathrm{r}}={\Sigma}\times \mathcal{X}$, where $\mathcal{X}$ is the Riemann sphere of ${\zeta}$ and ${\Sigma}$ is a two-dimensional 
complex conformal manifold (the complexification of $\Sigma$ in the previous section). The double fibration then becomes
\begin{equation*}
\begin{xy}
  \xymatrix@C=0.7cm@R=0.7cm{
      & \mathcal{F}_{\mathrm{r}} \ar[ld]_q \ar[rd]^p		&  \\
      {\Sigma}	 	&	&	\mathcal{R}     
  }
\end{xy}
\end{equation*}
Furthermore, we assume that on $Σ$ we are given a holomorphic solution $r$ of the Laplace equation and by $z$ we shall denote the harmonic conjugate of $r$.  The \textit{reduced twistor space} $\mathcal{R}$ associated with ${\Sigma}$ and $r$ is constructed from $\mathcal{F}_{\mathrm{r}}$ by identifying two points of $\mathcal{F}_{\mathrm{r}}$, $(σ,{\zeta})$ and $(σ',{\zeta}')$, if they lie on the same connected component of one of the surfaces given by
\begin{equation} \label{eq:quadw}
r{\zeta}^{2}+2(w-z){\zeta}-r=0
\end{equation}
for some value of $w$ and where $z=z(σ)$, $r=r(σ)$.

We can use $w$ as a local holomorphic coordinate on $\mathcal{R}$ and one can show that $\mathcal{R}$ is a non-Hausdorff Riemann surface covering $\mathbb{CP}^{1}$. As above, $w$ corresponds to one point of $\mathcal{R}$ if one can continuously change the roots of \eqref{eq:quadw} into each other by going on a path in ${\Sigma}$ and keeping $w$ fixed; and two points otherwise. Note that the condition for $w$ to correspond only to one point is an open condition in $\mathcal{R}$ as $z$ and $r$ are smooth functions on ${\Sigma}$. Now let $S$ be the $w$ Riemann sphere, and $V$ be the set of values for $w$ which correspond only to one point in $\mathcal{R}$. Then $V\subset S$ is open, and if ${\Sigma}$ is simply connected, then
\begin{equation} \label{eq:eqV}
V=\{z({\sigma})+\irm r({\sigma}): {\sigma}\in {\Sigma}\}. 
\end{equation}
In general, $V$ is not connected. However, for many solutions $J'$ is smooth and non-degenerate at the axis beyond $V$. Thus the region where the two spheres are identified can be enlarged to be simply connected such that this identification also extends to the fibres of $\hat E$.
\begin{Def}
An Ernst potential $J'$ is called \textit{axis-regular} if the corresponding bundle $E' \to  \mathcal{R}_{V}$ satisfies $E' = η^{*}(\hat E)$ where $\hat E$ is a bundle over $\mathcal{R}'=\mathcal{R}_{V'}$ such that $\hat E|_{S_{0}}$ and $\hat E|_{S_{1}}$ are trivial.

Here $\mathcal{R}'$ is a double cover of the $w$-Riemann sphere identified over the two copies of the set $V'$ where $V'$ is open, connected, simply connected, invariant under $w \mapsto  \bar w$ and $V \subseteq  V'$ ($V$ as in \eqref{eq:eqV}). The map $η: \mathcal{R}_{V}\to \mathcal{R}_{V'}$ is the projection.

We shall also say that a metric $J$ is axis-regular if the corresponding Ernst potential $J'$ is.
\end{Def}
Choose the copies of the two Riemann spheres in $\mathcal{R}'$ such that $\infty _{0}\in S_{0}$ and $\infty _{1}\in S_{1}$, see Figure~\ref{fig:redtwis}. 

If the bundle was not axis-regular, it meant that there are more than only isolated points where the two spheres cannot be identified. Thus, in the light of later results (Proposition~\ref{prop:AnalyCont}, Corollary~\ref{cor:singofP}) and \cite[App.~F]{Harmark:2004rm} axis-regularity is necessary for the space-time not to have curvature singularities at $r=0$.

\begin{figure}[htbp]
\begin{center}
     \scalebox{0.7}{\input{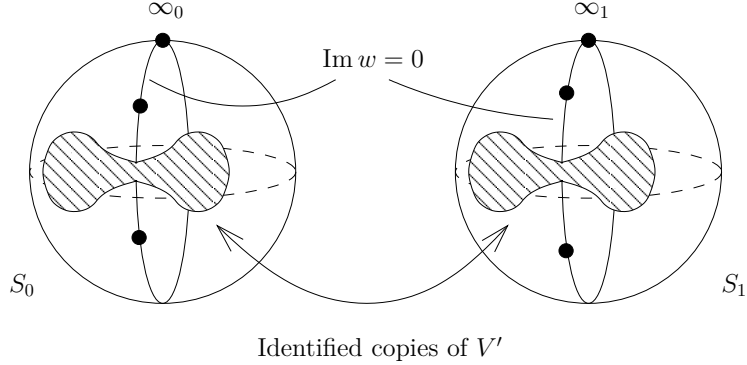}}
     \caption{Non-Hausdorff reduced twistor space with real poles (bullet points) as in the relevant examples.} 
     \label{fig:redtwis}
\end{center}
\end{figure}
Note that, whatever $x$ and $r$ are, $w=\infty $ will always correspond to two points, ${\zeta}=0$ and ${\zeta}=\infty $. Therefore, $w=\infty $ is never in $V'$ and we denote the points in $\mathcal{R}$ corresponding to $w=\infty $ by $\infty _{0}$ (for ${\zeta}=0$) and $\infty _{1}$ (for ${\zeta}=\infty $). The copies of the two Riemann spheres are chosen such that $\infty _{0}\in S_{0}$ and $\infty _{1}\in S_{1}$.  

In due course we will see that for the relevant examples there will altogether only be a finite number of such double points on the real axis $\{\infty ,w_{1},\ldots ,w_{n}\}$, which are the roots of
\begin{equation} \label{eq:rootswi}
r{\zeta}^{2}+2(w_{i}-z){\zeta}-r=0.
\end{equation}
In order to make the projection $p:\mathcal{F}_{\mathrm{r}} → \mathcal{R}$ well-defined we have to assign the roots ${\zeta}_{i}^{0}$, ${\zeta}_{i}^{1}$ to $S_{0}$, $S_{1}$, say ${\pi}({\zeta}_{i}^{0})\in S_{0}$ and ${\pi}({\zeta}_{i}^{1})\in S_{1}$.

Now we need to construct the reduced form of the Penrose-Ward transform. First the forward direction, where we are given $J$ as a solution of \eqref{eq:redyang}. In the unconstrained case the bundle over twistor space is defined by specifying what the fibres are. Here the holomorphic bundle $E\to \mathcal{R}$ is constructed by taking the fibre over a point of $\mathcal{R}$ to be the space of solutions of 
\begin{equation*}
\begin{split}
(\partial _{r}-{\sigma}\partial _{z}+r^{-1}{\sigma}\partial _{{\sigma}})s-{\sigma}(J^{-1}J_{z})s & = 0,\\
(\partial _{z}+{\sigma}\partial _{r}-r^{-1}{\sigma}^{2}\partial _{{\sigma}})s+{\sigma}(J^{-1}J_{r})s & = 0,
\end{split}
\end{equation*}
on the corresponding connected surface in $\mathcal{F}_{\mathrm{r}}$. The integrability condition for those two equations is precisely \eqref{eq:redyang}. 

Conversely, let $E\to \mathcal{R}$ be a holomorphic rank-$n$ vector bundle together with a choice of frame in the fibres. For a fixed ${\sigma}\in {\Sigma}$ let ${\pi}:\mathcal{X}\to \mathcal{R}$ be the restricted projection of $\mathcal{F}_{\mathrm{r}}\to \mathcal{R}$ to $\{{\sigma}\}\times \mathcal{X}$, that is the identification $({\sigma},{\zeta}) \sim ({\sigma}',{\zeta}')$ as above, and denote by ${\pi}^{*}(E)$ the pullback bundle of $E$ to $\mathcal{F}_{\mathrm{r}}$. We have to assume that ${\pi}^{*}(E)$ is a trivial holomorphic bundle over $\mathcal{X}$. This is less restrictive than it seems since if it is satisfied at one point ${\sigma}$ then it holds in a neighbourhood of ${\sigma}$.

The matrix $J$ can then be recovered within $J\mapsto AJB$, where $A$ and $B$ are constant, by a splitting procedure as follows. Suppose $E$ is  given by patching matrices $\{P_{{\alpha}{\beta}}(w)\}$ according to an open cover $\{\mathcal{R}_{{\alpha}}\}$ of $\mathcal{R}$ such that $\infty _{0}\in \mathcal{R}_{0}$ and $\infty _{1}\in \mathcal{R}_{1}$. Then ${\pi}^{*}(E)$ is given by patching matrices
\begin{equation} \label{eq:patmat1}
P_{{\alpha}{\beta}}(w({\sigma}))=P_{{\alpha}{\beta}}\left(\frac{1}{2}r({\sigma})({\zeta}^{-1}-{\zeta})+z({\sigma})\right)
\end{equation}
according to the open cover $\{{\pi}^{-1}(\mathcal{R}_{{\alpha}})\}$ of $\mathcal{X}$. The triviality assumption implies that there exist splitting matrices $f_{{\alpha}}({\zeta})$ such that
\begin{equation}\label{eq:patmat2}
P_{{\alpha}{\beta}}\left(\frac{1}{2}r({\sigma})({\zeta}^{-1}-{\zeta})+z({\sigma})\right)=f_{{\alpha}}^{\vphantom{-1}}({\zeta})f_{{\beta}}^{-1}({\zeta}).
\end{equation}
We define $J\coloneqq f_{0}(0)f_{1}(\infty )^{-1}$. Another splitting would be of the form $f_{{\alpha}}C$ for an invertible matrix $C$, which has to be holomorphic on the entire ${\zeta}$ Riemann sphere, thus $C$ is constant. But this leaves $J$ invariant and the definition is independent of the choice of splitting. The splitting matrices $f_{{\alpha}}$ depend smoothly on $r$, $z$ as ${\sigma}$ varies, so $J$ does. Although $J$ might have singularities where the triviality condition does not hold. It can be shown that a so-defined $J$ does indeed satisfy \eqref{eq:redyang}. 

It is important to note that, since the construction of $J$ relied on $π$, a different assignment of the roots of \eqref{eq:rootswi} to the Riemann spheres yields a different solution $J$. Yet, these different solutions are analytic continuations of each other and one can show that they are different parts of the Penrose diagram of the maximal analytic extension of the metric.

For the rest of this section we review results in the case $n-2=2$, that is $J$ is a $2\times 2$-matrix and $E$ a rank-2 vector bundle. An understanding of this is important for the generalization to higher dimensions. 

To characterize the bundle $E\to \mathcal{R}'$ in terms of patching matrices we choose a four-set open cover $\{U_{0},\ldots ,U_{3}\}$ of $\mathcal{R}'$ such that $U_{0}\cup U_{2}\supset S_{0}$ with $V'\subset U_{2}$ and $\infty _{0}\in U_{0}$, and $U_{1}\cup U_{3}\supset S_{1}$ with $V'\subset U_{3}$ and $\infty _{1}\in U_{1}$. Now we use the following theorem.
\begin{thm}[Grothendieck]
Let $E\to \mathbb{C}\mathbb{P}^{1}$ be a rank-$a$ vector bundle. Then
\begin{equation*}
E=L^{k_{1}}\oplus \ldots \oplus L^{k_{a}}=\mathcal{O}(-k_{1})\oplus \ldots \oplus \mathcal{O}(-k_{a})
\end{equation*}
for some integers $k_{1},\ldots ,k_{a}$ unique up to permutation. Here, $L^{k_{i}}=L^{\otimes k_{i}}$ with $L$ the tautological bundle.
\end{thm}
Hence we can choose a trivialization such that $\left.E\right|_{S_{0}}=L^{p}\oplus L^{q}$ and $\left.E\right|_{S_{1}}=L^{p'}\oplus L^{q'}$, that is 
\begin{equation*}
\renewcommand{\arraystretch}{1.5}
P_{02}=\left(\begin{array}{cc}(2w)^{p} & 0 \\0 & (2w)^{q}\end{array}\right),\quad 
P_{13}=\left(\begin{array}{cc}(2w)^{p'} & 0 \\0 & (2w)^{q'}\end{array}\right), 
\end{equation*}
where we assume that without loss of generality $\{w=0\}\subset V'$ which can be achieved by adding a real constant to $w$. Now the triviality assumption and the symmetry imply that $p=-p'$ and $q=-q'$.

That reduces the patching data to two integers $p$, $q$ and a single holomorphic patching matrix $P(w)=P_{23}(w)$ defined for $w\in V'$. The remaining patching matrices are obtained by concatenation. 

From the interpretation of $J$ as the matrix of inner products of Killing vectors in general relativity, we require $J$ to be real and symmetric. It is not hard to see that $J$ is symmetric, if and only if $P$ is symmetric. Furthermore, $J$ is real if and only if $P$ is real in the sense that $\overline{P(w)}=P(\skew{1}{\bar}{w})$.

Moreover, for $n=2$ it can be shown that (see for example \cite{Woodhouse:1988ek}, but as cited here it is taken from \cite{Fletcher:1990db}):
\begin{itemize}
\item If $\det P=1$, then $\det J=\left(-r^{2}\right)^{p+q}$. 
\item If $J$ is obtained from an axis-regular space-time and if the definition of ${\pi}$ is such that ${\zeta}_{i}^{0}\to 0$ and ${\zeta}_{i}^{1}\to \infty $ for $r\to \infty $ and all $i$, then $p=1$, $q=0$ and $P(z)=J'(0,z)$ on the rotational axis or on the horizon. Here, $J'$ is the Ernst potential. Thus, $P$ is the analytic continuation of the boundary values of the Ernst potential.
\item If $J$ comes from an asymptotically flat space-time in the sense that its Ernst potential has the same asymptotic form as the Ernst potential of Minkowski space with rotation and translation as Killing vectors, then $P(\infty )=1$, and conversely.
\end{itemize}

These results can be used to look at an example.
\begin{ex}[The Kerr solution]
The patching data for the Kerr solution is (without proof)
\begin{equation*}
\renewcommand{\arraystretch}{1.5}
P(w)=\frac{1}{w^{2}-{\sigma}^{2}}\left(\begin{array}{cc}
(w+m)^{2}+a^{2} & 2am \\2am & (w-m)^{2}+a^{2}
\end{array}\right)
\end{equation*}
where ${\sigma}=\sqrt{m^{2}-a^{2}}$ for $a<m$. Axis-regularity implies $p=1$ and $q=0$. The open set $V'$ is the complement of $\{\infty ,w_{1},w_{2}\}$ where $w_{1}={\sigma}$ and $w_{2}=-{\sigma}$. 
\end{ex}
\section{Black Holes and Rod Structure} \label{sec:bh}
In this section we are going to specify the assumptions about our space-time and define the rod structure, a tool which turns out useful for the black hole classification in higher dimensions.

\subsection{Black Holes --- Assumptions}

We assume our space-time to be asymptotically flat, stationary and axisymmetric. The axisymmetry is a $\mathrm{U}(1)$-symmetry, so we imagine it as a rotation around a codimension-2 hypersurface. 
Note that for $\dim M>4$ there is the possibility to rotate around multiple independent planes (if say $M$ is asymptotically flat). For spatial dimension $n-1$ we can group the coordinates in pairs $(x_{1},x_{2}), (x_{3},x_{4}), \ldots$ where each pair defines a plane for which polar coordinates $(r_{1},\varphi_{1}), (r_{2},\varphi_{2}), \ldots$ can be chosen. Thus there are $N=\lfloor \frac{n-1}{2}\rfloor$ independent (commuting) rotations each associated with an angular momentum.
\begin{Def}
An $n$-dimensional space-time $M$ will be called \textit{stationary and axisymmetric} if it admits $n-3$ of the above $\mathrm{U}(1)$ axisymmetries in addition to the timelike symmetry.
\end{Def}
However, note that this yields an important limitation. As shown in \cite[Sec.~3.1]{Myers:1986aa} and \cite{Emparan:2008aa}, for globally asymptotically flat space-times we have by definition an asymptotic factor of $S^{n-2}$ in the spatial geometry, and $S^{n-2}$ has isometry group $\mathrm{O}(n-1)$. The orthogonal group $\mathrm{O}(n-1)$ in turn has Cartan subgroup $\mathrm{U}(1)^{N}$ with $N=\lfloor \frac{n-1}{2}\rfloor$, that is there cannot be more than $N$ commuting rotations. But each of our rotational symmetries must asymptotically approach an element of $\mathrm{O}(n-1)$ so that $\mathrm{U}(1)^{n-3}\subseteq\mathrm{U}(1)^{N}$, and hence
\begin{equation*}
n-3\leq N=\left\lfloor \frac{n-1}{2}\right\rfloor,
\end{equation*}
which is only possible for $n=4,5$. Therefore, stationary and axisymmetric solutions in our sense can only have a globally asymptotically flat end in dimension four and five. 

Despite this, much of the theory is applicable in any dimension greater than four. The definitions can in that case be modified, for example by requiring that spacelike infinity is diffeomorphic to $\mathbb{R}^{n}\backslash B(R)\times N$ instead of diffeomorphic to $\mathbb{R}^{n}\backslash B(R)$, with $N$ a compact manifold of the relevant dimension.

Henceforth we are going to assume that, if not mentioned differently, we are given a \textit{vacuum (non-degenerate black hole) space-time $(M,g)$ which is five-dimensional, globally hyperbolic, asymptotically flat, 
stationary and axisymmetric, and that is analytic up to and including the boundary $r=0$. We are not considering space-times where there are points with a discrete isotropy group.} Note that stationarity and 
axisymmetry includes orthogonal transitivity, which was necessary for the construction in Section~\ref{sec:yang}. The assumption of analyticity might seem unsatisfactory, but in this paper we are going to 
focus on concepts concerning the uniqueness of five-dimensional black holes rather than regularity. The assumption of nondegeneracy, that is that horizons are nondegenerate, ensures that poles in the patching matrix $P(z)$ introduced below 
are simple. It may be possible to treat degenerate black holes by allowing double poles but we have not explored that.

\subsection{Rod Structure} \label{sec:rodstr}

Examples in five dimensions, such as the Myers-Perry solution and the black ring, have shown that the obvious analogue of the Carter-Robinson Theorem does not hold in higher dimensions, that is 
in five dimensions the mass and the two angular momenta are not sufficient to characterise black hole space-times. As extra variables the rod structure has been studied, for 
example in \cite{Hollands:2008fp}, with promising results.

In terms of the $(r,z)$-coordinates from Section~\ref{sec:yang} we define, as in \cite[Sec.~III.B.1]{Harmark:2004rm}.
\begin{Def} \label{def:rodstr}
A \textit{rod structure} is a subdivision of the $z$-axis into a finite number of intervals where to each interval a three-vector is assigned. 
The intervals are referred to as \textit{rods}, the vectors as \textit{rod vectors} and the finite number of points defining the subdivision as \textit{nuts}.
\end{Def}
In order to assign a rod structure to a given space-time we quote the following proposition.
\begin{prop}[Proposition~3 in \cite{Hollands:2008fp}]
Let $(M, g_{ab})$ be the exterior of a stationary, asymptotically flat, analytic, five-dimensional vacuum black hole space-time with connected horizon and isometry group $G=\mathrm{U}(1)^{2}\times \mathbb{R}$. Then the orbit space $\hat M = M / G$ is a simply connected 2-manifold with boundaries and corners. If $\skew{7}{\tilde}{A}$ denotes the matrix of inner products of the spatial (periodic) Killing vectors then furthermore, in the interior, on the one-dimensional boundary segments (except the segment corresponding to the horizon), and at the corners $\skew{7}{\tilde}{A}$ has rank 2, 1 or 0, respectively.
\end{prop}
Furthermore, since $\det \skew{7}{\tilde}{A} \neq  0$ in the interior of $\hat M$, the metric on the quotient space must be Riemannian. Then $\hat M$ is an (orientable) simply connected two-dimensional analytic manifold with boundaries and corners. The Riemann mapping theorem thus provides a map of $\hat M$ to the complex upper half plane where some further arguments show that the complex coordinate can be written as ${\zeta}=z+ \irm r$. So, starting with a space-time $(M,g)$ the line segments of the boundary $\partial  \hat M$ give a subdivision of the $z$-axis 
\begin{equation} \label{eq:sptrodsdef}
(-\infty , a_{1}),\ (a_{1},a_{2}), \ldots  , (a_{N-1},a_{N}),\ (a_{N},\infty )
\end{equation}
as the boundary of the complex upper half plane. This subdivision is moreover unique up to translation $z\mapsto z+\text{const.}$ which can be concluded from the asymptotic behaviour. 
For details see \cite[Sec.~4]{Hollands:2008fp}. This subdivision is now our first ingredient for the rod structure assigned to $(M,g)$. Since we have seen in Section~\ref{sec:bundles} that one nut is always at 
$z=∞$ (Corollary~\ref{cor:singofP}), we take the set of nuts to be $\{a_{0}=∞,a_{1},…,a_{N}\}$.

As the remaining ingredient we need the rod vectors. The imposed constraint $-r^{2}=\det J(r,z)$ implies $\det J(0,z)=0$, and therefore
\begin{equation*}
\dim \ker J(0,z)\geq 1.
\end{equation*}
We will refer to the set $\{r=0\}$ as the \textit{axis}. Taking the subdivision \eqref{eq:sptrodsdef} we define the 
rod vector for a rod $(a_{i},a_{i+1})$ as the vector that spans $\ker J(0,z)$ for $z\in (a_{i},a_{i+1})$ (we will not distinguish between the vector and its $\mathbb{R}$-span). A few comments on that. 

First consider the horizon. Assuming that our space-time is not static, we learn from the Rigidity Theorem \cite[Thm.~2]{Hollands:2007aa} that there exist $N$, $N\geq 1$, linear independent Killing vectors $X_{1},\ldots ,X_{N}$ that commute mutually and with the timelike Killing vector ${\xi}$. These Killing vector fields generate periodic commuting flows, and there exists a linear combination
\begin{equation*}
K={\xi}+{\Omega}_{1}X_{1}+\ldots +{\Omega}_{N}X_{N}, \quad {\Omega}_{i}\in \mathbb{R}, \qquad \text{(I)}
\end{equation*}
so that the Killing field $K$ is tangent and normal to the null generators of the horizon $\mathcal{H}$, and $g(K,X_{i})=0$~(II) on $\mathcal{H}$. These conditions are equivalent to
\begin{align*}
g_{ti} + \sum _{j}{\Omega}_{j}g_{ij} & = 0 \quad \text{on }\mathcal{H}, \qquad \text{(II)} \\
g_{tt} + 2 \sum _{i}{\Omega}_{i}g_{ti} +\sum _{i,j}{\Omega}_{i}{\Omega}_{j}g_{ij} & = 0 \quad \text{on }\mathcal{H}, \qquad \text{(I)}\\
\stackrel{\text{(II)}}{\Longleftrightarrow} g_{tt} + \sum _{i} {\Omega}_{i} g_{ti} & = 0 \quad \text{on }\mathcal{H}. 
\end{align*}
Hence
\begin{equation*}
\renewcommand{\arraystretch}{1.5}
J\tilde K = \left(\begin{array}{c}g_{tt} + \sum _{i} {\Omega}_{i} g_{ti} \\g_{ti} + \sum _{j}{\Omega}_{j}g_{ij}\end{array}\right) = 0  \quad \text{on }\mathcal{H},\ \text{where}\ \tilde K = \left(\begin{array}{c}1 \\{\Omega}_{1} \\{\Omega}_{2}\end{array}\right).
\end{equation*}
In other words $\tilde K$ is an eigenvector of $J$ on $\mathcal{H}$. So, by the change of basis ${\xi}\mapsto K$, $X_{i}\mapsto X_{i}$ the first row and column of $J$ diagonalizes with vanishing 
eigenvalue towards $\mathcal{H}$. On the other hand away from any of the rotational axes the axial symmetries $X_{1}$, $X_{2}$ are independent and non-zero, thus the rank of $J$ drops on the horizon precisely by one and the kernel is spanned by $\tilde K$. Note that if the horizon is connected precisely one rod in \eqref{eq:sptrodsdef} will correspond to $\mathcal{H}$. 

Second, consider the rods which do not correspond to the horizon (assuming that $\mathcal{H}$ is connected). Proposition 1 and the argument leading to Proposition 3 in \cite{Hollands:2008fp} show that on those rods the 
rotational Killing vectors are linearly dependent and the rank of $J$ again drops precisely by one. Whence, on each rod $(a_{i}, a_{i+1})$ that is not the horizon, there is a vanishing linear combination $aX_{1}+bX_{2}$. 
Therefore the vector $\left(\begin{array}{ccc}0 & a & b\end{array}\right)^{\mathrm{t}}$ spans the $\ker J(0,z)$, $z\in (a_{i}, a_{i+1})$. By \cite[Prop.~1]{Hollands:2008fp} $a$ and $b$ are constant so that we take $aX_{1}+bX_{2}$ as the rod vector on $(a_{i}, a_{i+1})$.

\begin{rem}
The fact that $a$ and $b$ are constant is not explicitly shown in the proof of \cite[Prop.~1]{Hollands:2008fp}, but follows quickly from \cite[Eq.~(11)]{Hollands:2008fp}. 
\end{rem}

Note that the nuts of the rod structure are the points which correspond to the corners of $\hat M$ and that is where the rank of $J$ drops precisely by two. So, at those points $\dim \ker J = 2$. 

\begin{ex}[Rod Structure of Four-Dimensional Schwarzschild Space-Time, taken from Section~3.1 in \cite{Fletcher:1990aa}] \label{ex:schwarz}
The Schwarzschild solution in four dimensions has in usual coordinates the form
\begin{equation*}
\drm s^{2} = -\left(1-\frac{2m}{R}\right) \drm T^{2}+\left(1-\frac{2m}{R}\right)^{-1} \drm R^{2} + R^{2}(\drm {\Theta}^{2}+ \sin^{2}{\Theta} \,\drm {\Phi}^{2}),
\end{equation*}
and is obtained in Weyl coordinates $(t,r,{\varphi},z)$ by replacing
\begin{equation*}
z=(R-m)\cos {\Theta}, \quad r=(R^{2}-2mR)^{\frac{1}{2}}\sin {\Theta},\quad t=T, \quad {\varphi}= {\Phi}.
\end{equation*}
If the symmetry group is generated by $X=\partial _{{\varphi}}$ and $Y=\partial _{t}$, then one can calculate the matrix of inner products of the Killing vectors as
\begin{equation*}
\renewcommand\arraystretch{1.8}
J=\left(\begin{array}{cc}\dfrac{r^{2}}{f} & 0 \\ \hphantom{-}0 & -f\end{array}\right)
\end{equation*}
where
\begin{equation*}
f=\frac{r_{+}+r_{-}-2m}{r_{+}+r_{-}+2m}\quad  \text{ with }\quad  r^{2}_{\pm}=r^{2}+(z\pm m)^{2}.
\end{equation*}
Note that $r_{+}=|z+m|$, $r_{-}=|z-m|$ for $r=0$ so that for $-m\leq z\leq m$ and $r=0$ we have $r_{+}=z+m$, $r_{-}=m-z$. Hence, $f$ vanishes for $r=0$, $-m\leq z\leq m$. Yet, applying l'H\^opital's rule twice shows that $r^{2}/f$ does not vanish for $r=0$, $-m<z<m$. So, the rod structure can be read off. It consists of the subdivision of the $z$-axis into $(-\infty ,-m)$, $(-m,+m)$ and $(+m,+\infty )$ and the rod vectors as in Figure~\ref{fig:SchWrodstr}. The semi-infinite rods correspond to the rotation axis and the finite one to the horizon. At $\{r=0, z=\pm m\}$ the entry $r^{2}/f$ blows up. Furthermore, we see that the boundary values of the rods are related to the mass of the black hole.\\ \mbox{} \hfill $\blacksquare$
\begin{figure}[htbp]
\begin{center}
     \scalebox{0.8}{\input{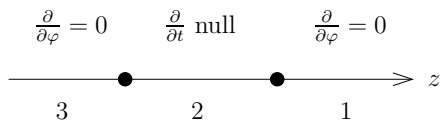}}
     \caption{Rod structure of the four-dimensional Schwarzschild solution. The numbers are only for the ease of reference to the parts of the axis later on.} 
     \label{fig:SchWrodstr}
\end{center}
\end{figure}
\end{ex}

There is a better way of visualizing the topology associated with the rod structure in five dimensions (from private communication with Piotr Chru\'{s}ciel). First consider five-dimensional Minkowski space. We leave the time coordinate and only focus on the spatial part. It is Riemannian, has dimension four, thus we can write it in double polar coordinates $(r_{1},{\varphi}_{1},r_{2},{\varphi}_{2})$. Then the first quadrant in Figure ~\ref{fig:Minkrodstr}, that is $\{r_{1}\geq 0,r_{2}\geq 0\}$, corresponds to the space-time.
\begin{figure}[htbp]
\begin{center}
     \scalebox{0.6}{\input{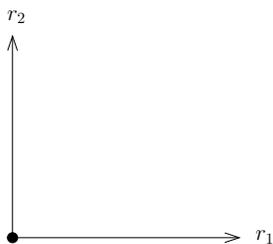}}
     \caption{Rod Structure for five-dimensional Minkowski space.} 
     \label{fig:Minkrodstr}
\end{center}
\end{figure}
The diagram suppresses the angles, so that each point in $\{r_{1}\geq 0,r_{2}\geq 0\}$ represents $S^{1}\times S^{1}$ where the radius of the corresponding circle is $r_{i}$. On the axes it thus degenerates to $\{\text{pt}\}\times S^{1}$. The boundary of our space-time, $r=0$, is in these polar coordinates $\{r_{1}=0\}\cup \{r_{2}=0\}$, and the nut is at the origin $r_{1}=r_{2}=0$.

Since our interest lies in asymptotically flat space-times, the rod structures for other space-times will be obtained from this one by modifying its interior and leaving the asymptotes unchanged. For example we can cut out a quarter of the unit disc as in Figure~\ref{fig:MProdstr}.
\begin{figure}[htbp]
\begin{center}
     \scalebox{0.6}{\input{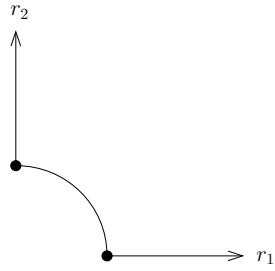}}
     \caption{Rod Structure with horizon topology $S^{3}$.} 
     \label{fig:MProdstr}
\end{center}
\end{figure}
But cutting out the quarter of the unit disc is nothing else than cutting out $r_{1}^{2}+r_{2}^{2}\leq 1$ (obviously taking the radius not to be one does not make any difference for the topology). Therefore the middle rod is the boundary of a region with topology $S^{3}$. So, if this is the horizon of a black hole, then the black hole has horizon topology $S^{3}$. Finally look at Figure~\ref{fig:BRrodstr1}.
\begin{figure}[htbp]
\begin{center}
     \scalebox{0.6}{\input{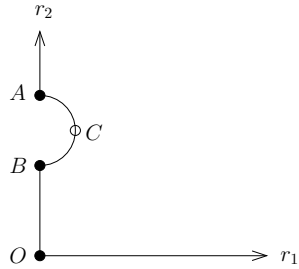}}
     \caption{Rod Structure with horizon topology $S^{2}\times S^{1}$. The nuts are at $A$, $B$ and $O$, where the rod between $A$ and $B$ corresponds to the horizon.} 
     \label{fig:BRrodstr1}
\end{center}
\end{figure}
This rod structure has three nuts:\;at $A$, at $B$ and at the origin $O$, that is at $r_{1}=r_{2}=0$. If the rod limited by $A$ and $B$ represents the horizon then the horizon topology is $S^{2}\times S^{1}$, which can be seen by rotating Figure~\ref{fig:BRrodstr1} first about the vertical and then about the horizontal axis. Another visualization is depicted in Figure~\ref{fig:BRrodstr2}, 
\begin{figure}[htbp]
\begin{center}
     \scalebox{0.6}{\input{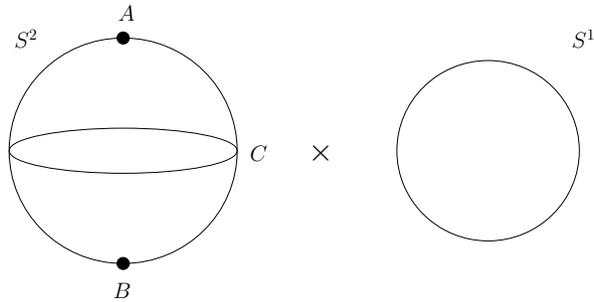}}
     \caption{Visualization of $S^{2}\times S^{1}$ topology.} 
     \label{fig:BRrodstr2}
\end{center}
\end{figure}
where the labelled points correspond to
\begin{align*}
A:&\ \{\text{pt}\}\times S^{1} \in  \mathbb{R}^{2}\times \mathbb{R}^{2},\ B:\ \{\text{pt}\}\times S^{1} \in  \mathbb{R}^{2}\times \mathbb{R}^{2}, \ O:\ \{\text{pt}\}\times \{\text{pt}\} \in  \mathbb{R}^{2}\times \mathbb{R}^{2},\\  C&:\ \hphantom{\{}S^{1} \hspace{0.1cm}\times S^{1} \in  \mathbb{R}^{2}\times \mathbb{R}^{2}.
\end{align*}
Interpolating the transition between those points explains the topology as well.

The relation between those more geometrical diagrams and the above definition of rod structure, that is only the $z$-axis with the nuts, can be made by the Riemann mapping theorem (see for example \cite[Sec.~4]{Hollands:2008fp}).

The rod structure is essential for the characterization of stationary axisymmetric black hole solutions. As mentioned above, such a solution in five dimensions is no longer uniquely given by its mass and angular momenta. But it is shown in \cite{Hollands:2008fp} that two solutions with connected horizon are isometric if their mass, angular momenta and rod structures coincide, if the exterior of the space-time contains no points with discrete isotropy group (see theorem in \cite{Hollands:2008fp}). 
\section{Twistor Approach in Five Dimensions} \label{sec:fivedim}
Most of what we have seen about the twistor construction at the end of Section~\ref{sec:bundles} generalizes at once to five and higher dimensions. Only, instead of two, 
the rank of the bundle for 5-dimensional vacuum solutions will be three so that we have three integers instead of only $p$ and $q$ in our twistor data (respectively $n-2$ integers in higher dimensions). The splitting 
procedure itself is not affected by increasing the rank. However, since the splitting itself is complicated, we have seen that the Ernst potential $J'$ is a crucial tool for 
any practical application of the twistor characterization of stationary axisymmetric solutions of Einstein's field equations, and the way we obtained $J'$ in 
Section~\ref{sec:bundles} seemed to be tailored to four dimensions with two Killing vectors. So, in order to pursue this strategy we have to define an Ernst 
potential in five dimensions. First, we are going to say a few words about B\"acklund transformations, because we will see that we secretly used them to obtain $J'$. 
Since some of the following extends immediately to higher dimensions as well, we will present most of it in $n$ dimensions. 

In the last part of this section we generalize results from \cite[Sec.~2.4]{Fletcher:1990aa} in order to conclude the important fact that the integers in the twistor 
data are non-negative as in four dimensions.

\subsection{B\"acklund Transformations}

In Section~\ref{sec:yang} we derived Yang's equation as one way of writing the ASDYM equations with gauge group $\mathrm{GL}(n,\mathbb{C})$. It has a number of `hidden' 
symmetries one of which is the B\"acklund transformation. 

As in \cite[Sec.~4.6]{Mason:1996hl} we can decompose a generic $J$-matrix in the following way
\begin{equation} \label{eq:bdecomp} 
\renewcommand{\arraystretch}{1.5}
J=
\left(\begin{array}{cc}A^{-1}-\skew{4}{\tilde}{B} \skew{7}{\tilde}{A} B & -\skew{4}{\tilde}{B} \skew{7}{\tilde}{A} \\ \skew{7}{\tilde}{A} B & \skew{7}{\tilde}{A}\end{array}\right)
=
\left(\begin{array}{cc}1 & \skew{4}{\tilde}{B} \hphantom{^{-1}} \\0 & \skew{7}{\tilde}{A}^{-1}\end{array}\right)^{-1}
\left(\begin{array}{cc}A^{-1} & 0 \\B \hphantom{^{-1}} & 1\end{array}\right),
\end{equation}
where $A$ is a $k\times k$ non-singular matrix ($k<n$), $\skew{7}{\tilde}{A}$ is $\skew{2}{\tilde}{k} \times  \skew{2}{\tilde}{k}$ non-singular matrix with 
$k+\skew{2}{\tilde}{k}=n$. Then, $B$ is a $\skew{2}{\tilde}{k}\times k$ and $\skew{4}{\tilde}{B}$ a $k\times \skew{2}{\tilde}{k}$ matrix. The term `generic' 
rules out for example cases where $\skew{7}{\tilde}{A}$ is not invertible. Substituting this in the reduced Yang's equation~\eqref{eq:redyang} we get the coupled 
system of equations\begin{equation} \label{eq:redbtrf} 
\begin{split}
r ∂_{z}(\skew{7}{\tilde}{A} B_{z}A)+∂ _{r}(r\skew{7}{\tilde}{A} B_{r}A) =0,\\
r∂_{z}(A \skew{4}{\tilde}{B}_{z}\skew{7}{\tilde}{A})+∂ _{r}(rA \skew{4}{\tilde}{B}_{r}\skew{7}{\tilde}{A}) =0,\\
r∂_{z}(\skew{7}{\tilde}{A}^{-1}\skew{7}{\tilde}{A}_{z})\skew{7}{\tilde}{A}^{-1}-∂ _{r}(r\skew{7}{\tilde}{A}^{-1}\skew{7}{\tilde}{A}_{r})\skew{7}{\tilde}{A}^{-1}+rB_{r}A\skew{4}{\tilde}{B}_{r}-rB_{r}A\skew{4}{\tilde}{B}_{r} =0,\\
rA^{-1}∂ _{z}(A_{z} A^{-1})-A^{-1}∂ _{r}(rA_{r} A^{-1})+r\skew{4}{\tilde}{B}_{r}\skew{7}{\tilde}{A} B_{r}-r\skew{4}{\tilde}{B}_{r}\skew{7}{\tilde}{A} B_{r} =0, 
\end{split}
\end{equation}
where an index denotes a partial derivative and all matrices are functions of $z$ and $r$. The first two equations are integrability conditions and they imply 
the existence of matrices $B'$ and $\skew{4}{\tilde}{B}'$ such that
\begin{align*}
∂_{r} \skew{4}{\tilde}{B}' & = r \skew{7}{\tilde}{A} B_{z}A, ∂_{z} \skew{4}{\tilde}{B}' & = -r\skew{7}{\tilde}{A} B_{r}A,\\
∂_{r} B' & = rA\skew{4}{\tilde}{B}_{z} \skew{7}{\tilde}{A}, ∂_{z} B' & = -rA\skew{4}{\tilde}{B}_{r} \skew{7}{\tilde}{A}.
\end{align*}
\begin{Def}
Together with $B'$ and $\skew{4}{\tilde}{B}'$ we define the other primed quantities as 
\begin{equation*} 
\left(A,\skew{7}{\tilde}{A},B,\skew{4}{\tilde}{B},k,\skew{2}{\tilde}{k}\right) \mapsto \left(A'=\skew{7}{\tilde}{A}^{-1},\skew{7}{\tilde}{A}'=r^{-2}A^{-1},B',\skew{4}{\tilde}{B}',k'=\skew{2}{\tilde}{k},\skew{2}{\tilde}{k}'=k\right).
\end{equation*}
We call the matrix $J'$, that is obtained from $J$ by \eqref{eq:bdecomp} with the primed blocks instead of the unprimed, the \textit{B\"acklund transform}.
\end{Def}
\begin{prop}[Section~4.6 in \cite{Mason:1996hl}] \label{prop:btsolyang}\mbox{}
\begin{enumerate}
\item $J'$ is again a solution of the reduced Yang's equation.
\item The B\"acklund transformation is invertible.
\end{enumerate}
\end{prop}
\begin{proof}\mbox{}
\begin{enumerate}
\item Substitute the unprimed by the primed versions in \eqref{eq:redbtrf}.
\item Noting that $B_{w}^{\vphantom{1}}=\skew{7}{\tilde}{A}^{-1}\skew{4}{\tilde}{B}'_{\tilde z}A^{-1}=A'\skew{4}{\tilde}{B}'_{\tilde z}\skew{7}{\tilde}{A}$ and 
similar for the other integrability equations we see that the definition of $B'$ and $\skew{4}{\tilde}{B}'$ is involutive. The inverse transformation on the other 
blocks is easy to see.
\end{enumerate}
\end{proof}
\begin{prop}\label{prop:bdet} \mbox{}
\begin{equation*}
\det J' = (-r)^{2(1-k)}
\end{equation*}
\end{prop}
\begin{proof} 
Note that for the decomposition of a general matrix in block matrices we know from basic linear algebra 
\begin{equation*}
\det \left(\begin{array}{cc}P & Q \\R & S\end{array}\right)
= \det(S) \det(P-QS^{-1}R),
\end{equation*}
if $S$ is invertible. Applied to our decomposition \eqref{eq:bdecomp} this yields
\begin{equation} \label{eq:detj} 
\det J = \det(\skew{7}{\tilde}{A}) \det(A^{-1}-\skew{4}{\tilde}{B} \skew{7}{\tilde}{A} B + \skew{4}{\tilde}{B} \skew{7}{\tilde}{A} \skew{7}{\tilde}{A}^{-1}\skew{7}{\tilde}{A} B)=\det(\skew{7}{\tilde}{A})\det(A^{-1}).
\end{equation}
Now using the fact that $\det J =- r^{2}$ we obtain
\begin{align*} 
\det J' & = \det (\skew{7}{\tilde}{A}') \det ((A')^{-1}) = (-r)^{-2k} \det (A^{-1}) \det (\skew{7}{\tilde}{A})= (-r)^{-2k} \det (J)\\
	& = (-r)^{2(1-k)}
\end{align*}
\end{proof}

\begin{ex}[B\"acklund Transformation in Four Dimensions]
We consider the four-dimensional case with two Killing vectors. As in \eqref{eq:backldec4d} we set $A=r^{-2} f$, $\skew{7}{\tilde}{A} = -f$, $\skew{4}{\tilde}{B}=-B={α}$. 
Then the integrability equations imply the existence of a function $ψ$ as in \eqref{eq:intweyl}. The blocks for the B\"acklund transform are $B'=-\skew{4}{\tilde}{B}'={ψ}$, 
$A'=\skew{7}{\tilde}{A}^{-1}=f^{-1}$, $\skew{7}{\tilde}{A}'=-r^{-2}A^{-1}=f^{-1}$ defining $J'$ as in \eqref{eq:ernstpot4d}. Hence, our Ernst potential in four dimensions is 
obtained by a B\"acklund transformation.\\ \mbox{} \hfill $\blacksquare$
\end{ex}

\subsection{Higher-Dimensional Ernst Potential}

Let us first recall the definition of twist 1-forms, twist potentials and some of their properties.
\begin{Def}
Consider an $n$-dimensional (asymptotically flat) space-time $M$ with $X_{0}$ a stationary and $X_{1},\ldots ,X_{n-3}$ axial Killing vectors, all mutually commuting. 
The \textit{twist $1$-forms} are defined as
\begin{align*}
ω_{1 a}^{\vphantom{1}} & = Δ\,{ε}_{ab\ldots cde}^{\vphantom{1}}X_{1}^{b}\cdots X_{n-3}^{c}∇ ^{d}X_{1}^{e}, \\
& \hspace{0.2cm} \vdots \\
ω_{n-3, a}^{\vphantom{1}} & =Δ\,{ε}_{ab\ldots cde}^{\vphantom{1}}X_{1}^{b}\cdots X_{n-3}^{c}∇ ^{d}X_{n-3}^{e},
\end{align*}
where  $Δ=\sqrt{-g}=r{\erm}^{2{\nu}}$, according to \eqref{eq:sigmametric}. (Note that in \cite{Hollands:2008fp} the notation is taken from \cite{Wald:1984rz} where ${ε}$ 
is already the volume element.)
\end{Def}
Adopting a vector notation
\begin{equation*}
ω=\left(\begin{array}{c}ω_1 \\ \vdots \\ω_{n-3}\end{array}\right), \quad X=\left(\begin{array}{c}X_1 \\ \vdots \\X_{n-3}\end{array}\right),
\end{equation*}
this can be written as
\begin{equation*}
ω_{a}^{\vphantom{1}} = Δ\,{ε}_{ab\ldots cde}^{\vphantom{1}}X_{1}^{b}\cdots X_{n-3}^{c}∇ ^{d}X^{e}. 
\end{equation*}
Define ${\theta}_{I}^{\vphantom{j}}$, $I\in \{1,\ldots ,n-3\}$, by ${\theta}_{Ik}^{\vphantom{j}}=g_{kj}^{\vphantom{j}}X_{I}^{j}$. 
\begin{prop} \mbox{}
\begin{enumerate}
\item The twist 1-forms can be written as
\begin{equation} \label{eq:twisthodge}
ω_{I}=* ({\theta}_{1}∧ \ldots ∧ {\theta}_{n-3}∧ \drm {\theta}_{I}), \quad I\in \{1,\ldots ,n-3\}.
\end{equation}
\item $ω$ is closed.
\item $ω$ annihilates the Killing vector fields $X_{0},\ldots ,X_{n-3}$.
\end{enumerate}
\end{prop}
\begin{proof} \mbox{}
\begin{enumerate}
\item This is clear.
\item This is implied by the vacuum field equations, analogously to the proof of \cite[Thm.~7.1.1]{Wald:1984rz}.
\item Using \eqref{eq:twisthodge}, a direct calculation shows that $ω_{r}$ and $ω_{z}$ are the only non-vanishing components of the twist 1-forms, from which this follows. \qedhere
\end{enumerate}
\end{proof}
The last statement in the above proposition implies that $ω$ can also be regarded as a set of 1-forms on the interior of $M / \mathcal{G}$. Due to the form of the ${\sigma}$-model 
metric there should be no confusion if we denote both the form on $M$ and the one on $M / \mathcal{G}$ by the same symbol. Being a form on $M / \mathcal{G}$ means that $ω$ 
has only non-vanishing components for the $r$- and $z$-coordinate, and of course as a form on $M / \mathcal{G}$ it will again be closed. This leads to the following 
definition. 
\begin{Def}
Locally there exist functions on $M / \mathcal{G}$ such that 
\begin{equation*}
∂ _{r} χ_{I}=ω_{Ir} \text{ and } ∂ _{z} χ_{I}=ω_{Iz}\quad  \text{for } I=1,\ldots ,n-3,
\end{equation*}
or equivalently in vector notation
\begin{equation*}
\drm χ = ω.
\end{equation*}
These functions $χ_{I}$ are called \textit{twist potentials}. 
\end{Def}
The construction of the Ernst Potential in \eqref{eq:ernstpot4d} is tailor-made for dimension four, and it is not immediately obvious how to generalize it to higher 
dimensions. Nevertheless, there is an ansatz in \cite{Maison:1979aa}, where it is noted that the full metric on space-time, that is essentially $J$, can be reconstructed 
from knowing the two twist potentials (in five dimensions) $χ$, the $2×2$-matrix $\skew{7}{\tilde}{A}=\left(X_{I}^{a}X_{K}^{b}g_{ab}\right)_{I,K=1,2}$ and its non-vanishing 
determinant $\det \skew{7}{\tilde}{A}$ on the factor space $M / \mathcal{G}$. The matrix in \cite[Eq.~(16)]{Maison:1979aa} will then be our candidate for the higher-dimensional 
Ernst Potential. Note, however, that the condition $\det \skew{7}{\tilde}{A}≠0$ needs further investigation. 

Closely connected to the non-vanishing determinant is the concept of adaptations to certain parts of the axis $r=0$. Recall that the assumption of axis-regularity was an 
important one. It said that the region where the two spheres of the reduced twistor space are identified can be enlarged to a simply connected patch such that this identification 
also extends to the fibres of the bundle. The exact shape of $V'$ is not important, however, there is still an ambiguity if we have a nut on $r=0$ (remember that we assume that 
there is only a finite number of isolated nuts). Figure~\ref{fig:adaptation} shows how we can choose different extensions of $V$.
\begin{figure}[htbp]
\begin{center}
     \scalebox{0.6}{\input{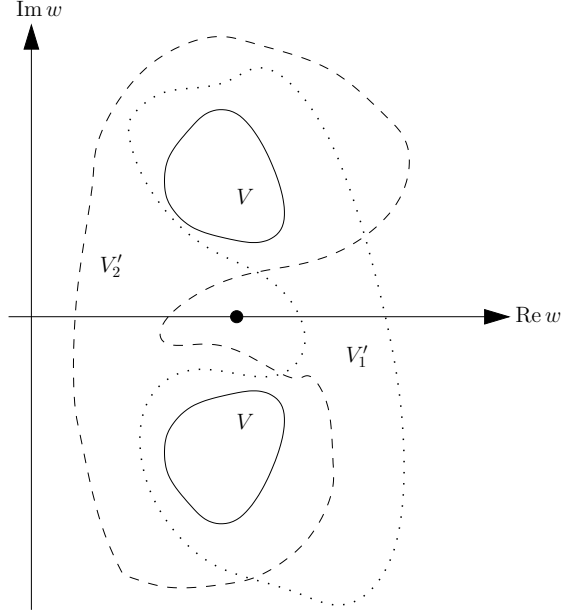}}
     \caption{Two different extensions $V'_{1}$ and $V'_{2}$ of $V$ around a pole of $J'$ (bullet on the real axis).} 
     \label{fig:adaptation}
\end{center}
\end{figure}
From Example~\ref{ex:schwarz} we learn that the choice of $V'$ matters, that is we obtain different Ernst potentials $J'$ for different extensions.
\begin{Def} 
In any dimension given an axis-regular $J'$ we shall call it \textit{adapted to the rod $(a_{i},a_{i+1})$} if $V'|_{r=0}\subseteq (a_{i},a_{i+1})$.
\end{Def}
For a rod corresponding to a rotational axis we know that along this rod a linear combination of the axial Killing vectors vanishes. By a change of basis we can always 
assume that without loss of generality this vector is already in the basis, say $X_{K}=0$, $K>0$. Then we make the following definition.
\begin{Def}
In dimension $n$ we call \textit{$\skew{7}{\tilde}{A}$ adapted to $(a_{i},a_{i+1})$} the $(n-3)×(n-3)$-matrix that is obtained from $J$ by 
\begin{enumerate}[(a)]
\item cancelling the $K^{\mathrm{th}}$ column and row, if $(a_{i},a_{i+1})$ is a rotational axis and $X_{K}$ the corresponding rod vector; 
\item cancelling the $0^{\mathrm{th}}$ column and row, if $(a_{i},a_{i+1})$ is the horizon.
\end{enumerate}
\end{Def}
The following lemma shows the reason for the latter definition.
\begin{lem}
$\skew{7}{\tilde}{A}$ adapted to $(a_{i},a_{i+1})$ is invertible on $(a_{i},a_{i+1})$ and becomes singular at the limiting nuts.
\end{lem}
\begin{proof}
This is a simple consequence of what we have seen in Section~\ref{sec:rodstr}. First, consider a rod corresponding to a rotational axis $(a_{i},a_{i+1})$. Here $J$ is an 
$(n-2)\times (n-2)$ matrix and has rank $n-3$. Since $X_{K}=0$ on $(a_{i},a_{i+1})$, we cancelled a zero column and row, thus it follows that $\skew{7}{\tilde}{A}$ has full 
rank and $\det \skew{7}{\tilde}{A}\neq 0$ on $(a_{i},a_{i+1})$. Of course, where the rank of $J$ drops further, that is where $\dim \ker J (0,z)\geq 2$, the matrix 
$\skew{7}{\tilde}{A}$ cannot have full rank anymore and $\det \skew{7}{\tilde}{A}=0$. So, this adaptation becomes singular as soon as we reach one of the nuts limiting 
this rod. 

Second, consider the horizon rod $(a_{h},a_{h+1})$. Here the first row and column of $J$, that is the one with the asymptotically timelike Killing vector in it, becomes zero. 
Then $\skew{7}{\tilde}{A}$ adapted to $(a_{h},a_{h+1})$ has full rank on $(a_{h},a_{h+1})$, and becomes singular at the nuts where the rotational axes intersect the horizon. 
\end{proof}

Note that at the beginning of this section we defined the twist 1-forms by singling out $X_{0}$. However, the definition works just as well with the set 
$X_{0},\ldots ,\hat X_{K},\ldots ,X_{n-2}$ of Killing vectors. For an adaptation to a certain rod we obtain in the same way as above $n-3$ twist potentials by omitting 
the rod vector for this rod. 

Now, we have collected all the components for the following definition, which is a modification of \cite[Eq.~(16)]{Maison:1979aa}.
\begin{Def} \label{def:ernstpot}
In $n$ dimensions and for a given rod $(a_{i},a_{i+1})$ we call the matrix
\begin{equation} \label{eq:highernst}
\renewcommand{\arraystretch}{1.5}
J'=\frac{1}{\det \skew{7}{\tilde}{A}} \left(\begin{array}{cc}\hphantom{-}1 & -χ^{\mathrm{t}} \\-χ & \det \skew{7}{\tilde}{A} \cdot \skew{7}{\tilde}{A} + χχ^{\mathrm{t}}\end{array}\right),
\end{equation}
where $\skew{7}{\tilde}{A}$ is adapted to $(a_{i},a_{i+1})$ and $χ$ is the vector of twist potentials for $(a_{i},a_{i+1})$, \textit{higher-dimensional Ernst potential adapted 
to $(a_{i},a_{i+1})$}.
\end{Def}
\begin{rem}\mbox{}
\begin{enumerate}
\item Evidently, $J'$ is symmetric and $\det J'=1$.
\item For $n=5$ this becomes on the horizon rod
\begin{equation} \label{eq:ernstn5} 
\renewcommand{\arraystretch}{1.5}
J'= \frac{1}{\det \skew{7}{\tilde}{A}} \left(\begin{array}{ccc}
\hphantom{-}1 & -χ_{1} & -χ_{2} \\
-χ_{1}^{\vphantom{1}} & \det \skew{7}{\tilde}{A} \cdot J_{11}^{\vphantom{1}}+χ_{1}^{2} & \det \skew{7}{\tilde}{A} \cdot J_{12}^{\vphantom{1}}+χ_{1}^{\vphantom{1}}χ_{2}^{\vphantom{1}} \\
-χ_{2}^{\vphantom{1}} & \det \skew{7}{\tilde}{A} \cdot J_{21}^{\vphantom{1}}+χ_{2}^{\vphantom{1}}χ_{1}^{\vphantom{1}} & \det \skew{7}{\tilde}{A} \cdot J_{22}^{\vphantom{1}}+χ_{2}^{2}
\end{array}\right).
\end{equation}
\item Note that $J=\big(g(X_{i},X_{j})\big)$ is a matrix of scalar quantities, hence $J$ is bounded for $r\to 0$. So, the domain of $J'$ is only restricted by 
$\det \skew{7}{\tilde}{A}$ and by the arguments above we see that for an adaptation to $(a_{i},a_{i+1})$ the limit $J'(0,z)$ is well-defined for $z∈(a_{i},a_{i+1})$.
\end{enumerate}
\end{rem}
Even though both its ingredients and the matrix $J'$ itself were known, the crucial new step for the twistor construction is to recognize the following.
\begin{thm}
$J'$ is obtained from $J$ by a B\"acklund transformation.
\end{thm}
\begin{proof}
Show by direct calculation with the help of Laplace expansions that $A = - r^{-2} \det \skew{7}{\tilde}{A}$, and then $B' = -\skew{4}{\tilde}{B}' = χ$.
\end{proof}
This completes the justification for calling $J'$ defined in \eqref{eq:highernst} the Ernst potential --- it is symmetric and has unit determinant and by Proposition~\ref{prop:btsolyang} it satisfies Yang's equation.

\subsection{The Limit towards the Axis}

Note that if $P$ has a pole with $r≠0$, then it is obviously not splittable at this point (see \eqref{eq:patmat2}), hence we do not obtain $J$ at this point by the 
splitting procedure, which means the metric will be singular at this point. We shall exclude such situations. However, referring again to \eqref{eq:patmat2} 
it is evident that the splitting procedure in general does break down for $r→0$. So, one has to study this limit by means of other tools and we will see later on that
 even for regular space-times $P$ has real poles. As in Section~\ref{sec:bundles} we assume that there are only finitely many of them.

In \cite[Sec.~2.4]{Fletcher:1990aa} an asymptotic formula for $J$ as $r\to 0$ was given (originally, the result in four dimensions goes back to \cite{Ward:1983yg}). The 
method used there can be extended to higher dimensions, but since the intricacy of the extension lies rather in its technicality than in a deeper idea, we will only 
state the result here and refer for the proof to \cite[Sec.~8.4]{Metzner:2012aa}. Consider space-time dimension $n$ and assume that without loss of generality the integers 
in the twistor data are ordered, $p_0≥…≥p_{n-3}$. Then near the axis, that is for $r→0$, we obtain
\begin{equation} \label{eq:Jasympt}
J→\left(\begin{array}{@{}c@{\hspace{-0.2mm}}c@{\hspace{-0.2mm}}c@{}}r^{p_0} &  &  \\ & \ddots &  \\ &  & r^{p_{n-3}}\end{array}\right)
\left(\begin{array}{@{}cc@{}}g & -g Υ^{\mathrm{t}} \\-g Υ \rule{0cm}{5mm} & g Υ Υ^{\mathrm{t}} - \skew{7}{\tilde}{\mathcal{A}}\end{array}\right)
\left(\begin{array}{@{}c@{\hspace{-2mm}}c@{\hspace{-2mm}}c@{}}(-1)^{p_0} &  &  \\ & \ddots &  \\ &  & (-1)^{p_{n-3}}\end{array}\right)
\left(\begin{array}{@{}c@{\hspace{-0.2mm}}c@{\hspace{-0.2mm}}c@{}}r^{p_0} &  &  \\ & \ddots &  \\ &  & r^{p_{n-3}}\end{array}\right),
\end{equation}
where 
\begin{equation*}
Υ=\left(\begin{array}{c}Υ^{(1)} \\ \vdots \\Υ^{(n-3)}\end{array}\right) \quad  \text{with } Υ^{(l)}=\frac{r^{p_{0}-p_{l}}}{2^{p_{0}-p_{l}}(p_{0}-p_{l})!}χ^{(p_{0}-p_{l})}(z),
\end{equation*}
and $\skew{7}{\tilde}{\mathcal{A}}(w)$ a matrix whose entries are derived from the splitting of the patching matrix into a part holomorphic in $ζ$, that is on 
$S_{0}^{\vphantom{-1}}$, and a part holomorphic in $ζ^{-1}$, that is on $S_{1}^{\vphantom{-1}}$. 

This formula implies that on a section of the axis $r=0$ where $A$ is not singular, that is everywhere apart from the nuts, because of the boundedness of $J$ we must have 
$p_{i}\geq 0$. But then $\det J = -r^{2}$ leads to
\begin{equation*}
p_{0}=1,\ p_{1}=0, \ldots ,\ p_{n-3}=0, 
\end{equation*}
and if $\det J =1$ as for the Ernst potential we get back the analytic continuation formula
\begin{equation*}
J(r,z) = P(z) + \text{ higher orders in } r.
\end{equation*}

\section{Properties of $P$ and the Bundle} \label{sec:adapt}

Having different adaptations of $J'$ which are related to different extensions of $V$, the question arises whether these yield equivalent bundles over the reduced twistor 
space. To see that they do is not very hard using results in \cite{Fletcher:1990aa}.

\begin{prop}[Proposition~3.1 in \cite{Fletcher:1990aa}] \label{prop:EquBdls}
Suppose $E\to \mathrm{R}_{V}$, the bundle corresponding to a solution $J$ of Yang's equation on the set $V$, can be represented as the pullback of the bundles 
$E^{1}\to \mathrm{R}_{V'_{1}}$ and $E^{2}\to \mathrm{R}_{V'_{2}}$, where $V'_{1}$ and $V'_{2}$ are simply connected open sets which intersect the real axis in 
distinct intervals. Then $E^{1}$ and $E^{2}$ are the pullbacks of a bundle $\tilde E \to  \mathrm{R}_{V'}$ where $V'=V'_{1}\cup V'_{2}$. Moreover, we can express 
$\tilde E$ in standard form in two different ways; and one of the collections of patching matrices is identical to the collection used to describe $E^{1}$ and the 
other collection is identical to that used to describe $E^{2}$.
\end{prop}
\begin{proof}
The proof in \cite{Fletcher:1990aa} only makes use of the construction for the reduced twistor space and not of the rank of the bundle. So, it carries over to higher
 dimensions.
\end{proof}

The bridge between Ernst potential and patching matrix is built by analytic continuation. 
\begin{prop}[Proposition~7.2 in \cite{Woodhouse:1988ek}] \label{prop:AnalyCont}
In five as well as in four dimensions if $P$ is a patching matrix of an axis-regular Ernst potential $J'$ on $V$, then $J'$ is analytic on (a choice of) $V'$ and 
$J'(0,z)=P(z)$ for real $z$. 
\end{prop}
\begin{proof}
This is literally the same as for \cite{Woodhouse:1988ek}, since it does not make use of the rank of the bundle. Note that even though in the statement there the 
assumption of $J$ being positive definite is made, it in fact is not necessary for the proof.
\end{proof}

\begin{cor} \label{cor:singofP}
A patching matrix $P$ has real singularities, that is points $z\in \mathbb{R}$ where an entry of $P$ has a singularity, at most at the nuts of the rod structure.
\end{cor}
\begin{proof}
Suppose $P$ corresponds to the bundle $\tilde E \to  \mathcal{R}_{\tilde V}$, where $\tilde V$ is the maximally extended region over which the spheres can be 
identified. Suppose further that $P$ has a real singularity $\skew{2} \tilde a \in \mathbb{R}$ which is not one of the nuts, say without loss of generality 
$\skew{2} \tilde a\in (a_{i},a_{i+1})$. Then from Proposition~\ref{prop:EquBdls} we know that $\tilde E$ can on $V_{i}$ be expressed in standard form. But using 
Proposition~\ref{prop:AnalyCont} that means that $P(z)=J'_{i}(0,z)$ for $z\in (a_{i},a_{i+1})$. On the other hand we have seen earlier already that $J'_{i}$ 
is regular on $(a_{i},a_{i+1})$ and only becomes singular when approaching the nut. Contradiction!

For every other bundle $E→\mathcal{R}_{V}$ it must be that $V⊆\tilde V$, hence $E$ is the pullback of $\tilde E$ and as such the patching matrix of $E$ cannot have 
poles where the patching matrix of $\tilde E$ has not.
\end{proof}

As in \cite{Fletcher:1990aa} we sometimes also call the singularities of $P$ \textit{double points}, because the singularities of $P$ are the points where the two 
Riemann spheres of the reduced twistor space cannot be identified.

\begin{prop} \label{prop:simplepoles}
The real singularities of a patching matrix $P$ are simple poles.
\end{prop}
\begin{proof}

We have seen above that on the real axis $r=0$ the singularities of $P$ are caused by the term $\det \skew{7}{\tilde}{A}$. So, consider the rod $(a_{i},a_{i+1})$ 
where $\skew{7}{\tilde}{A}$ has full rank. The determinant of a matrix equals the product of its eigenvalues. Furthermore, towards the nuts $a_{i}$, $a_{i+1}$ we know, 
also from above, that the rank of $\skew{7}{\tilde}{A}$ drops precisely by one which is the case if and only if $\det \skew{7}{\tilde}{A}$ contains the factors $z-a_{i}$ 
and $z-a_{i+1}$, respectively, with multiplicity one. 
\end{proof}

We have seen earlier that the various adaptations of the $P$-matrix represent equivalent bundles. A tool that will be very useful for the reconstruction of the space-time
 from the given data is to know more precisely how to obtain the $P$-matrix adapted to one part of the axis from other adaptions. The following theorem relates the adaptions 
on the outer rods for five-dimensional space-times.

\begin{thm} \label{thm:invpmatrix}
In five space-time dimensions, if $P_{+}$ is the patching matrix adapted to $(a_{\scriptscriptstyle N}, ∞)$, then $P_{-}^{\vphantom{1}}=MP_{+}^{-1}M$ with 
\begin{equation*}
M=\left(\begin{array}{ccc}0&0&1\\0&1&0\\1&0&0\end{array}\right)
\end{equation*}
is the patching matrix adapted to $(-∞,a_{1})$.
\end{thm}
\begin{proof}
Remember that the map $π$ in Section~\ref{sec:bundles} is only well-defined if one specifies the assignment of the poles to the spheres. The roots of the double points 
$a_{i}$ satisfy
\begin{equation*}
rζ_{i}^{2} + 2(a_{i}^{\vphantom{2}}-z)ζ_{i}^{\vphantom{2}} -r =0,
\end{equation*}
so they are
\begin{equation} \label{eq:rootsdoublepoints}
ζ_{i}^{±} = \frac{1}{r} \left((a_{i}^{\vphantom{2}}-z)±\sqrt{(z-a_{i}^{\vphantom{2}})^{2}+r^{2}}\right).
\end{equation}
We note two things. First, the spheres are labelled by saying that the roots of $w=∞$, namely $ζ=0$ and $ζ=∞$, are mapped to $π(0)=∞_{0}∈S_{0}$ and $π(∞)=∞_{1}∈S_{1}$. 
Second, $r$ and $z$ are chosen as parameters in the very beginning, but the obtained expressions depend smoothly on $r$ and $z$ so that we can vary them and follow the 
consequences. One observation of this kind is that for $r→0$ one of the roots in \eqref{eq:rootsdoublepoints} tends to zero and one to infinity. 

Hence, given a solution $J$ of Yang's equation~\eqref{eq:redyang} and the corresponding bundle $E→\mathcal{R}_{U}$, $U=\mathbb{C}\mathbb{P}^{1}\backslash \{∞,a_{1},…,a_{\scriptscriptstyle N}\}$ the region where the spheres are identified and $U$ not simply connected, then the description of the twistor space as 
$S_{0}^{\vphantom{2}}∪S_{1^{\vphantom{2}}}$ and the patching matrix $P$ are adapted to the component $\mathcal{C}$ of the real axis if those $ζ_{i}^{±}$ that tend to 
zero for $r→0$ on $\mathcal{C}$ are assigned to $S_{0}^{\vphantom{2}}$ and those that tend to infinity are assigned to $S_{1}^{\vphantom{2}}$; see also \cite[Prop.~3.2]{Fletcher:1990aa}. This is merely a requirement of consistent behaviour under the variation of $r$ and $z$, since $π(0)∈S_{0}^{\vphantom{2}}$ and 
$π(∞)∈S_{1}^{\vphantom{2}}$. Note that in this case on $\mathcal{C}$ it is $P(z)=J'(z)$. 

More explicitly this can be stated as
\begin{equation*}
\renewcommand{\arraystretch}{1.3}
ζ_{i}^{+} → \left\{ \begin{array}{ll}
0, & i≤k \\
∞, & i>k
\end{array}\right. \quad \text{and} \quad
ζ_{i}^{-} → \left\{ \begin{array}{ll}
∞, & i≤k \\
0, & i>k
\end{array}\right. ,
\end{equation*}
for an adaptation to $\mathcal{C}=(a_{k},a_{k+1})$ and for $r→0$ on $z∈(a_{k},a_{k+1})$.

This allows us to draw the conclusion that for a given bundle a change of adaptation from $(a_{k},a_{k+1})$ to $(a_{k-1},a_{k})$ is achieved by swapping the 
assignment of $π(ζ_{k}^{±})$ to the spheres; see \cite[Sec.~3.2]{Fletcher:1990aa}. Following this idea, one then obtains the adaptation to $(-∞,a_{1})$  from 
an adaptation to $(a_{\scriptscriptstyle N}, ∞)$ by swapping all double points $π(ζ_{k}^{±})$, $1≤k≤N$. However, the latter is the same as swapping the double 
point at infinity which means relabelling the spheres by saying $π(0)∈S_{1}^{\vphantom{2}}$ and $π(∞)∈S_{0}^{\vphantom{2}}$.

Now let us step back from this line of ideas and have a look from another side. Note that if $J$ is a solution of Yang's equation~\eqref{eq:redyang} with 
$\det J =1$ and $J$ is symmetric, $J=J^{\mathrm{t}}$, then $J^{-1}$ is a solution of Yang's equation as well with $\det J^{-1}=1$. On the other hand, just 
by inspection of the splitting procedure one notices that $J$ is defined as a linear map $J:E_{∞_{0}}→E_{∞_{1}}$, where $E_{w}$ is the fiber of $E→\mathcal{R}$ 
over $w∈\mathcal{R}$, and that $J^{-1}$ is the solution of Yang's equation generated by the bundle where the spheres are swapped, $S_{0}↔S_{1}$ (see also property 
(3) in \cite[Sec.~4]{Woodhouse:1988ek}). But this is precisely what we have done above.

The last point to note is that even though we have shown that $P^{-1}$ is adapted to $(-∞,a_{1})$ it does not necessarily have to be in our standard form due to the 
gauge freedom in the splitting procedure. Analyzing the eigenvalues of \eqref{eq:ernstn5} we see that asymptotically, that is for $r=0$ and $z→\pm ∞$, the first one 
becomes infinite, the second one is bounded and the third one goes to zero. Taking the inverse of $J'$ adapted to the top rod swaps the behaviour of the first and 
third eigenvalue, thus $P^{-1}$ is brought in standard form by swapping the first and third row and column. This is implemented by the conjugation with $M$, which 
completes the proof.
\end{proof}
\begin{rem}
The last point in the proof above about bringing $P^{-1}$ into standard form will be more obvious when we have seen more details about the asymptotic behaviour in the 
second part of this paper.
\end{rem}
\begin{cor} \label{cor:invpmatrix}
In five space-time dimensions, if $P_{+}$ is the patching matrix adapted to $(a_{\scriptscriptstyle N}, ∞)$, then ${\Delta} \coloneqq\prod _{i=1}^{N}(z-a_{i})$ 
divides all $2×2$-minors of $Δ⋅P_{+}^{\vphantom{-1}}=P'_{+}$.
\end{cor}
\begin{proof}
From Theorem~\ref{thm:invpmatrix} and Proposition~\ref{prop:simplepoles} it follows that $P_{+}^{-1}$ has at most simple poles at the nuts. But by the general formula 
for the inverse of a matrix the entries of $P_{+}^{-1}$ are (up to a sign) $P_{+}^{i,j}/Δ^{2}$, where $P_{+}^{i,j}$ is the $2×2$-minor of $Δ⋅P_{+}$ obtained by cancelling 
the $i^{\mathrm{th}}$ row and $j^{\mathrm{th}}$ column. Hence one factor of $Δ$ has to cancel.
\end{proof}
\begin{rem}
Taking the Ernst potential \eqref{eq:highernst} in five dimensions and writing it in the following way
\begin{equation*}
\renewcommand{\arraystretch}{1.4}
J'_{+}(z)=\left(\begin{array}{cc} \hphantom{-}g & -gχ^{\mathrm{t}} \\-gχ & \skew{7}{\tilde}{A} + gχχ^{\mathrm{t}}\end{array}\right)
 = \frac{1}{Δ}\left(\begin{array}{cc}p_{0} & \vec{p}^{\,\mathrm{t}} \\ \vec{p} & \mathbb{P} \end{array}\right),
\end{equation*}
the matrix of metric coefficients $\skew{7}{\tilde}{A}$ as a function of $z$ is obtained as
\begin{equation}\label{eq:metricfromP}
\skew{7}{\tilde}{A} = \frac{1}{Δ} \mathbb{P} - \frac{1}{Δp_{0}}\vec{p}\cdot \vec{p}^{\,\mathrm{t}} = \frac{1}{Δp_{0}}\left(p_{0}\mathbb{P}-\vec{p}\cdot \vec{p}^{\,\mathrm{t}}\right).
\end{equation}
All entries of $p_{0}\mathbb{P}-\vec{p}\cdot \vec{p}^{\,\mathrm{t}}$ are $2×2$-minors of $Δ⋅J'$, hence $Δ$ divides them. Thus $\skew{7}{\tilde}{A}= \tilde{\mathbb{P}}/ p_{0}$ 
where the entries of $\tilde{\mathbb{P}}$ are polynomials in $z$. We remember from above that $p_{0}/Δ$ blows up when we approach $a_{N}$, that means $p_{0}$ cannot have a factor 
$(z-a_{N})$. So, the entries of $\skew{7}{\tilde}{A}$ are bounded as $z↓a_{N}$, a feature which is consistent with our picture of space-time.

Note, however, that we cannot extend that to other nuts without changing the adaptation, that is to say, the expression for a metric coefficient $J_{ij}(r=0,z)$, $z>a_{N}$, 
might contain poles for $z<a_{N}$. An example (which will be studied in more detail in Part~II of this article) is the black ring, where
\begin{align*}
\renewcommand{\arraystretch}{2.5}
J_{22}(z) = \left\{\begin{array}{cl}\dfrac{2(z-{\kappa}^{2})(z+c{\kappa}^{2})}{z-c{\kappa}^{2}} & ,\ z>{\kappa}^{2}, \\0 & ,\ c{\kappa}^{2}<z<{\kappa}^{2}, \\ \dfrac{2(z-c{\kappa}^{2})(z+c{\kappa}^{2})}{z-{\kappa}^{2}} & ,\ -c{\kappa}^{2}<z<c{\kappa}^{2}, \\0 & ,\ z<-c{\kappa}^{2},\end{array}\right.
\end{align*}
with $c$ and $κ$ parameters of the solution. The denominators vanish for certain values of $z$, but these are not singularities of the metric since they are 
not in the domain of the respective expression. Given this observation, it will be important not to impose too strong conditions when deducing the patching matrices in Part II of this article, 
because when we try to fix the free parameters we cannot take \eqref{eq:metricfromP} and say because the metric is regular the denominator has to divide the numerator up to 
a constant.

Note further that even though Theorem~\ref{thm:invpmatrix} and Corollary~\ref{cor:invpmatrix} generalize directly to $n$ dimensions, the conclusion for the metric 
coefficients in $\skew{7}{\tilde}{A}$ does not. This is because in higher dimensions the entries of $\skew{7}{\tilde}{A}$ will still consist of certain $2×2$-minors of 
$P_{+}^{\vphantom{1}}$ as in \eqref{eq:metricfromP}, whereas $P_{+}^{-1}$ being a patching matrix requires $Δ^{n-4}$ to divide the $(n-3)×(n-3)$-minors of 
$P_{+}^{\vphantom{1}}$. This coincides in five dimensions, but is not implied automatically for dimensions greater than five. Yet, the boundedness of the 
metric coefficients ought to hold always, so that it at most gives extra conditions on the free parameters.\\ \mbox{} \hfill $\blacksquare$
\end{rem}

\section{Summary and Outlook}
In this work we have reviewed the existing twistor construction for stationary, axisymmetric and asymptotically flat solutions of Einstein's equations in 
four dimensions and identified the points which inhibit the immediate generalization to higher dimensions. Of crucial importance was the Ernst potential, 
whose definition had to be extended to higher dimensions and it turned out that this is done by a B\"acklund transformation.

For the classification of five-dimensional black holes the rod structure plays a significant role. So, having obtained a generalized Ernst potential, we 
explained in detail how the rod structure and twistor data of a space-time correspond to each other. In a continuation of this work we will study, how 
these results can be used to reconstruct space-time metrics via the patching matrix from the given data, angular momenta and rod structure. For this 
reconstruction it is necessary to understand how the patching matrix changes when moving to an adjacent rod. A first step was done in Theorem~\ref{thm:invpmatrix} 
and Corollary~\ref{cor:invpmatrix}, which relate the patching matrices adapted to the outer intervals of the rod structure.

\hypersetup{
	bookmarksdepth=0
}

\section*{Acknowledgements}
The author would like to thank Paul Tod for fruitful discussions and his advice on this work. Furthermore, he is grateful to Lionel Mason, Piotr Chru\'{s}ciel and Nicholas Woodhouse for sharing their thoughts on some ideas. The author was supported by a PhD studentship from the German National Academic Foundation (Studienstiftung des deutschen Volkes), a Lamb and Flag Scholarship from St\,John's College Oxford and the EPSRC studentship MATH0809.

\nocite{Ward:1983yg}
\bibliographystyle{jphysicsB}						
\bibliography{/Users/norman/mathematics/Papers/library} 							
\end{document}